%% file: main.tex
\documentclass[11pt,a4paper]{article}

\usepackage[margin=1in]{geometry}
\usepackage[T1]{fontenc}
\usepackage{lmodern}
\usepackage{amsmath,amssymb,amsthm,mathtools}
\usepackage{microtype}
\usepackage{booktabs}
\usepackage{algorithm}
\usepackage{algpseudocode}
\usepackage{flafter}
\usepackage[numbers,sort&compress]{natbib}
\usepackage[hidelinks]{hyperref}
\usepackage[capitalise,nameinlink,noabbrev]{cleveref}
\makeatletter
\providecommand{\theHALG@line}{}
\renewcommand{\theHALG@line}{\thealgorithm.\arabic{ALG@line}}
\makeatother

\hypersetup{
  pdftitle={Randomized Online Fair Division: High-Probability and Expected Realized Fairness},
  pdfauthor={Tianqi Chen and Jingxiao Long},
  pdfsubject={High-probability and expected realized guarantees for randomized online goods allocation},
  pdfkeywords={online fair division, randomized algorithms, high-probability fairness,
    expected realized fairness, PROP1, EFX, EF1, non-adaptive adversary}
}
\title{Randomized Online Fair Division:\\High-Probability and Expected Realized Fairness}
\author{%
  Tianqi Chen\thanks{\fontsize{8.5}{10}\selectfont
    School of Mathematical Science, Zhejiang University,
    Hangzhou 310027, P. R. China. Email:
    \href{mailto:ctq@zju.edu.cn}{ctq@zju.edu.cn}.}
  \quad
  Jingxiao Long\thanks{\fontsize{8.5}{10}\selectfont
    School of Mathematical Science, Zhejiang University,
    Hangzhou 310027, P. R. China. Email:
    \href{mailto:jx_long@zju.edu.cn}{jx\_long@zju.edu.cn}.}%
}
\date{}

\newtheorem{theorem}{Theorem}[section]
\newtheorem{lemma}[theorem]{Lemma}
\newtheorem{proposition}[theorem]{Proposition}
\newtheorem{corollary}[theorem]{Corollary}
\theoremstyle{definition}

\theoremstyle{remark}

\newcommand{\Rand}{\textnormal{\textsc{Rand}}}
\newcommand{\OSQ}{\textnormal{\textsc{OSQ}}}
\newcommand{\E}{\mathbb E}
\newcommand{\Prb}{\mathbb P}
\newcommand{\ind}{\mathbf 1}

\begin{document}
\maketitle
\begin{abstract}
We study randomized algorithms for the fully online allocation of
indivisible goods among \(n\ge2\) agents with nonnegative additive
valuations. Goods arrive sequentially and must be allocated immediately
and irrevocably, with only \(n\) known in advance. Since exact ex-ante
envy freeness and proportionality are readily achievable, while no
positive ex-post approximation is possible for the fairness notions
considered here, we study the intermediate notions of high-probability
fairness and expected realized fairness.
Against a non-adaptive adversary, we give a randomized algorithm for
proportionality up to one good (PROP1) whose parameter depends only on
\(n\) and that preserves exact ex-ante envy-freeness and proportionality.
At confidence \(1-\delta\), its PROP1 guarantee improves on independent
uniform allocation (\(\Rand\)) by a factor of \(\Omega(\log n)\),
uniformly over \(\delta\in(0,1/2]\).
As \(n\to\infty\), its expected realized PROP1 factor is at
least \(\frac{3-\sqrt5}{2}-o(1)\). We also show that the expected realized PROP1 factor of \(\Rand\) is
\((1+o(1))/\log n\), yielding an improvement of at least
\(\bigl(\frac{3-\sqrt5}{2}-o(1)\bigr)\log n\) for our algorithm.
For every randomized online algorithm and every positive approximation factor, 
the success probability can be made arbitrarily small for envy freeness up to 
any good (EFX) and at most \(\frac{n+1}{2n}\) for envy freeness up to one good (EF1).
Consequently, every randomized fully online algorithm has an expected
realized EFX guarantee of zero and an expected realized EF1 guarantee
of at most \(\frac{n+1}{2n}\).
\end{abstract}


\input{sections/overview}
\input{sections/preliminaries}
\input{sections/prop1_osq}
\input{sections/impossibility}
\input{sections/conclusion}
\label{end:main}
\bibliographystyle{ACM-Reference-Format}
\bibliography{references}
\end{document}

%% file: sections/overview.tex
\section{Introduction}

Fair division studies how resources can be allocated among agents with
different preferences in a fair manner. In the classical offline model,
all resources and valuations are known before an allocation is chosen.
Many applications, however, require resources to be allocated as they
arrive, without knowledge of future arrivals. This gives rise to
\emph{online fair division} \cite{AleksandrovEtAl2015}, where each decision must be made using only the information revealed so far.

We study \emph{randomized algorithms} for the fully online allocation of
indivisible goods among \(n\ge2\) agents with nonnegative
additive valuations. Goods
arrive one by one, and every good must be allocated immediately and
irrevocably after its valuation vector is revealed. Nothing is known in advance except the number of agents. We consider a non-adaptive
adversary, so the entire input is fixed independently of the
algorithm's random choices.

Randomized allocations are commonly evaluated from both ex-ante and
ex-post perspectives \cite{AzizFreemanShahVaish2024}. Ex-ante fairness
is determined by the fractional allocation induced by the marginal
allocation probabilities, and is often easy to achieve. For example,
a rule that selects one agent uniformly at random and allocates every
good to that agent is exactly ex-ante envy-free and proportional, even
though every realized allocation can be highly unfair. Thus, ex-ante
fairness alone provides little control over realized outcomes.

At the other extreme, ex-post fairness requires every allocation in
the support of the randomized algorithm to satisfy the desired fairness
guarantee. This requirement is too strong in this model.
Recent work establishes that deterministic online algorithms
cannot guarantee any positive multiplicative approximation to 
proportionality up to one good (PROP1),
with the impossibility extending to a broad range of standard
envy-based, proportionality-based, and share-based fairness notions
\cite{NeohTeh2026}. As we show in Lemma~\ref{lem:expost-equivalence}, these
deterministic impossibility results also rule out the corresponding
ex-post guarantees for randomized algorithms, even against a non-adaptive
adversary. Hence ex-ante fairness
can be too weak to control realized outcomes, while ex-post fairness can
be too strong to achieve.

Motivated by this gap between ex-ante and ex-post fairness, we study
realized guarantees for PROP1, envy-freeness up to one good (EF1), and
envy-freeness up to any good (EFX). For a fairness notion \(F\), let
\(\rho_F\in[0,1]\) denote the approximation factor achieved by a realized
allocation. For PROP1, this factor measures each agent's value after
allowing one outside good relative to her proportional share. For EF1
and EFX, it measures the remaining envy after, respectively, removing
one suitably chosen good or any good from another agent's bundle.
We evaluate \(\rho_F\) through \emph{high-probability
fairness} and \emph{expected realized fairness}. A high-probability 
guarantee controls the lower tail of \(\rho_F\),
whereas an expected realized guarantee controls its expectation. Their
formal definitions and relation are given in
\cref{subsec:randomized-guarantees}.

For PROP1, independent uniform allocation achieves the tight
high-probability factor
\(\Theta(1/\log(n/\delta))\) for
\(\delta\in(0,1/2]\) against a non-adaptive adversary
\cite{ChooEtAl2026}. Neoh and
Teh~\cite[Section~6]{NeohTeh2026} ask whether a general randomized
online algorithm can improve on this guarantee for unrestricted
instances. For the stronger envy-based notions EFX and EF1, we ask
whether relaxing ex-post fairness to
high-probability or expected realized guarantees can recover any
non-trivial approximation in this model.
We answer both questions below.

\subsection{Our Contributions}

\paragraph{Improved guarantees for PROP1.}
We introduce the \emph{Ordered Scale Quota} algorithm (\(\OSQ\)), a
randomized online allocation rule whose group size depends only on \(n\).
It satisfies exact ex-ante envy-freeness and proportionality
(\cref{lem:osq-marginals}).

Against a non-adaptive adversary, for all sufficiently large \(n\),
\cref{thm:osq-high-probability} gives an explicit guarantee for every
fixed input and every failure probability \(\delta\in(0,1)\).
At the same confidence level \(1-\delta\), \(\OSQ\) improves the
high-probability PROP1 guarantee of \(\Rand\) by a factor
of \(\Omega(\log n)\), uniformly over \(\delta\in(0,1/2]\).

In particular, for every fixed \(c>0\), all sufficiently large \(n\),
and every fixed input \(I\) with \(n\) agents,
\[
 \Prb\left(\rho_{\mathrm{PROP1}}(\OSQ,I)
        \ge\frac{3-\sqrt5}{2}-o(1)\right)\ge1-n^{-c}.
\]
The same asymptotic lower bound holds for its expected
realized PROP1 guarantee:
\[
 R_{\mathrm{PROP1}}^{\OSQ}(n)
 \ge\frac{3-\sqrt5}{2}-o(1)
 \qquad\text{(\cref{cor:osq-expected})},
\]
where \(R_F^{\mathcal A}(n)\) denotes the expected realized
guarantee of algorithm \(\mathcal A\) for fairness notion~\(F\).
All error terms are independent of the input.

For the independent uniform rule \(\Rand\), we prove the asymptotically
sharp bound
\[
 R_{\mathrm{PROP1}}^{\Rand}(n)
 =\frac{1+o(1)}{\log n}
 \qquad\text{(\cref{prop:rand-prop1})}.
\]
Thus \(\OSQ\) improves the expected realized PROP1 guarantee
of \(\Rand\) by a factor of at least
\(\bigl(\frac{3-\sqrt5}{2}-o(1)\bigr)\log n\).

Our analysis combines a deterministic certificate based on the shared
scale records with an exponential-moment bound obtained by induction
along the fixed agent order.

\paragraph{Limitations for EFX and EF1.}
For every \(n\ge2\) and every randomized fully online algorithm,
the success probability of any positive EFX approximation factor can
be made arbitrarily small, even when all agents have identical
valuations (\cref{thm:efx-success}). Consequently,
\(
R_{\mathrm{EFX}}^{\mathcal A}(n)=0
\)
for every such algorithm \(\mathcal A\)
(\cref{cor:efx-expected}).

For EF1, no randomized online algorithm can guarantee a positive
approximation factor with confidence exceeding
\(\frac{n+1}{2n}\) (\cref{thm:ef1-success}). Moreover,
\(
R_{\mathrm{EF1}}^{\mathcal A}(n)
\le
\frac{n+1}{2n}
\)
for every such algorithm \(\mathcal A\)
(\cref{cor:ef1-expected}).

\paragraph{Paper organization.}
The remainder of the paper is organized as follows.
\Cref{sec:model} introduces the model and formalizes the fairness
criteria and randomized guarantees.
\Cref{sec:prop1} presents \(\OSQ\), establishes its PROP1 guarantees,
and compares them with those of \(\Rand\).
\Cref{sec:impossibility} proves the limitations for EFX and EF1.

\subsection{Related Work}

\paragraph{Online fair division.}
Aleksandrov et al.~\cite{AleksandrovEtAl2015} initiated the systematic
study of online fair division of indivisible goods in a food-bank
setting, where goods arrive sequentially and must be allocated upon
arrival. They introduced and analyzed simple online mechanisms,
including the Like and Balanced Like rules, from the perspectives of
fairness, incentives, and efficiency. Benad\`e et
al.~\cite{BenadeEtAl2018} subsequently studied quantitative envy in a
closely related online model. They showed that, although envy generally
cannot be kept bounded, its growth can be made sublinear, obtaining
asymptotically tight vanishing-envy guarantees. Aleksandrov and
Walsh~\cite{AleksandrovWalsh2020} survey this early literature and
organize the different models, mechanisms, and normative properties
that arise in online fair division.

More recent work studies multiplicative approximations to standard fair
division notions in the fully online setting.
Choo et al.~\cite{ChooEtAl2026} study approximate PROP1
guarantees. They show that several natural greedy algorithms fail to guarantee any
positive approximation against an adaptive adversary, while against a
non-adaptive adversary the independent uniform rule \(\Rand\) achieves
a meaningful high-probability guarantee. In
particular, they determine its tight asymptotic order
\(\Theta(1/\log(n/\delta))\) for
\(\delta\in(0,1/2]\). Neoh and Teh~\cite{NeohTeh2026} further show that
no positive multiplicative approximation to PROP\(k\), for any fixed
\(k\), is possible against an adaptive adversary, with the impossibility
extending to a broad range of envy-based, proportionality-based, and
share-based fairness notions, including EF1 and EFX. Under a
non-adaptive adversary, they also characterize the
high-probability PROP1 performance of the Like rule, which assigns each
good uniformly among the agents who value it positively. They further
ask whether a general randomized online algorithm can outperform
independent uniform allocation on unrestricted instances.
Our work answers this question by using
correlated randomization to obtain stronger realized PROP1 guarantees.

\paragraph{Ex-ante and ex-post fairness.}
A growing line of work studies randomized allocations that combine
ex-ante fairness with guarantees for every realized allocation.
Aziz et al.~\cite{AzizFreemanShahVaish2024} show that, for additive
valuations, one can simultaneously achieve exact ex-ante envy-freeness
and ex-post EF1. Babaioff et al.~\cite{BabaioffEzraFeige2022}
combine ex-ante proportionality with ex-post fair-share guarantees.
Feldman et al.~\cite{FeldmanEtAl2024} extend this best-of-both-worlds
approach to subadditive valuations, obtaining ex-ante
\(1/2\)-envy-freeness together with ex-post EF1 and
\(1/2\)-EFX.

These results are obtained in offline settings and require the desired
ex-post guarantee to hold for every allocation in the support. In the
fully online setting studied here, such support-wide guarantees can be
impossible. We instead evaluate the fairness of realized allocations
through high-probability and expected realized guarantees.

\paragraph{Additional information and alternative online models.}
A complementary line of work studies how additional information about
future goods changes the guarantees achievable online. One common model
assumes that each agent's total value for all goods is known in advance,
or equivalently that valuations are normalized
\cite{ZhouBaiWu2023,NeohPetersTeh2026,ChenTan2026}.
Other forms of partial information have also been considered.
Choo et al.~\cite{ChooEtAl2026} introduce predictions of each agent's
maximum item value (MIV), and subsequent work strengthens the resulting
PROP1 guarantees and studies the limits of exact MIV information
\cite{NeohTeh2026,DuSun2026}. Neoh et al.~\cite{NeohPetersTeh2026} also study frequency predictions, where
the multiset of values that each agent will assign to future goods is
known but their arrival order is not.

Temporal fair division can be viewed as the complete-information end of
this spectrum. The allocation remains sequential and irrevocable, but
the future sequence is known in advance and fairness is required for
the cumulative allocation at every prefix rather than only at the end.
Elkind et al.~\cite{ElkindEtAl2025} introduce this informed temporal
model for indivisible items and study temporal EF1. Related work
investigates per-round and cumulative fairness, as well as stronger
temporal notions such as EFX and MMS
\cite{CooksonEbadianShah2025,ChoiLi2026,ChoiLiTeh2026}.

Another way to obtain positive online guarantees is to restrict the
valuation domain. Amanatidis et al.~\cite{AmanatidisEtAl2025} study
personalized two-value instances, where each agent assigns one of two
agent-specific values to every good. Wang and Wei~\cite{WangWei2026}
study binary and personalized bi-valued valuations for goods and
chores, obtaining fairness and efficiency guarantees unavailable for
general additive valuations. 

Other work modifies the online allocation model itself. Reordering
buffers allow a limited number of arriving items to be stored and
reordered before allocation \cite{AmanatidisEtAl2026Buffer}.
Budget-constrained online fair division restricts which assignments are
feasible and evaluates fairness relative to budget-feasible bundles
\cite{CohenEtAl2026}. Online fair division with subsidies allows
monetary transfers to restore envy-freeness
\cite{KulkarniEtAl2026}. Related work also studies online fairness when
previous allocations may be revised through a bounded number of
reallocations \cite{HeEtAl2019}. 


%% file: sections/preliminaries.tex
\section{Model and Fairness Guarantees}\label{sec:model}

\subsection{Model and realized fairness scores}

There are \(n\ge2\) agents, indexed by \([n]=\{1,\ldots,n\}\), and a finite sequence
\(I=(g_1,\ldots,g_m)\) of indivisible goods.
Agent \(i\in[n]\) has a nonnegative additive valuation \(v_i\),
meaning that \(v_i(S)=\sum_{g\in S}v_i(\{g\})\) for every bundle \(S\).
We write \(v_i(g)\) as shorthand for \(v_i(\{g\})\).
On arrival, the vector \(v(g_t)=(v_i(g_t))_{i\in[n]}\) is revealed and
\(g_t\) must be assigned immediately to one agent. Previous assignments
cannot be changed, and goods cannot be discarded. Only \(n\) is known in 
advance, while the horizon \(m\), future valuation
vectors, and total values remain unknown.
A non-adaptive adversary fixes the complete input independently
of the algorithm's randomness.
Let \(M=\{g_1,\ldots,g_m\}\), \(V_i=v_i(M)\), and
\(A=(A_1,\ldots,A_n)\) denote a terminal allocation.
We use \(\mathcal A\) for a randomized online algorithm and
\(A^{\mathcal A}(I)\) for its random terminal allocation on input \(I\).
All probabilities and expectations are over the algorithm's internal
randomness, with the input held fixed.

For a fixed allocation \(A\), we measure its approximation factors for
PROP1~\cite{ConitzerFreemanShah2017}, EF1~\cite{Budish2011}, and
EFX~\cite{CaragiannisEtAl2019} by scores in \([0,1]\).
Define \(h(x,0)=1\) and \(h(x,d)=\min\left\{1,\frac{x}{d}\right\}\) for \(x\ge0,d>0\).
Writing \(U_i=v_i(A_i)\) and
\(b_i=\max(\{v_i(g):g\notin A_i\}\cup\{0\})\), put
\[
 \rho_{\mathrm{PROP1}}(A)=\min_i h\left(U_i+b_i,\frac{V_i}{n}\right).
\]
For \(\alpha\in[0,1]\), \(\rho_{\mathrm{PROP1}}(A)\ge\alpha\) means every agent reaches an
\(\alpha\) fraction of her proportional share after adding at most one
outside good.

For the envy-based notions \(F\in\{\mathrm{EF1},\mathrm{EFX}\}\), set
\[
 \rho_F(A)=\min_{i\ne j}h(U_i,d_{ij}^F),
\]
where, for a nonempty \(A_j\),
\[
 d_{ij}^{\mathrm{EF1}}=v_i(A_j)-\max_{g\in A_j}v_i(g),
 \qquad
 d_{ij}^{\mathrm{EFX}}=v_i(A_j)-\min_{g\in A_j}v_i(g).
\]
Both denominators are zero for an empty bundle.
EF1 permits deleting one good chosen for the comparison, whereas EFX requires the
comparison after every deletion, including a zero-valued good.
All scores belong to \([0,1]\), equal one
exactly when the corresponding fairness requirement holds.
For each \(F\in\{\mathrm{PROP1},\mathrm{EF1},\mathrm{EFX}\}\), the event
\(\rho_F(A)\ge\alpha\) is precisely simultaneous \(\alpha\)-\(F\)
fairness of the allocation.

\subsection{High-probability and expected realized guarantees}\label{subsec:randomized-guarantees}

For \(F\in\{\mathrm{PROP1},\mathrm{EF1},\mathrm{EFX}\}\), write
\[
 \rho_F(\mathcal A,I):=\rho_F\bigl(A^{\mathcal A}(I)\bigr).
\]
This is the realized score of the output, viewed as a random variable.

\paragraph{High-probability fairness.}
For a failure probability \(\delta\in(0,1)\), an algorithm has an
\(\alpha(n,\delta)\)-\(F\) guarantee with confidence \(1-\delta\) if,
for every fixed input \(I\),
\[
 \Prb_{\mathcal A}\bigl(\rho_F(\mathcal A,I)
                \ge\alpha(n,\delta)\bigr)\ge1-\delta,
\]
where \(\alpha(n,\delta)\in[0,1]\) is independent of \(I\).
This is one event in which all required agent or pair comparisons
hold simultaneously. Letting \(\delta\) tend to zero gives the usual
asymptotic high-probability interpretation.

\paragraph{Expected realized fairness.}
We evaluate an algorithm by its worst-case expected realized score:
\[
 R_F^{\mathcal A}(n)
 =\inf_I\E_{\mathcal A}\bigl[\rho_F(\mathcal A,I)\bigr].
\]
The infimum ranges over all finite inputs with \(n\) agents in the
model above. The minimum over agents or pairs is taken within each
realized allocation before expectation is taken.
For each \(n\), this criterion summarizes an algorithm's realized fairness
by a single expected factor, without requiring a choice of the failure
probability \(\delta\).

\paragraph{The tail-integral relation.}
Fix an algorithm and an input, and put
\(X=\rho_F(\mathcal A,I)\in[0,1]\). Its mean is the area under its
success-probability curve:
\[
 \E[X]
 =\int_0^1\Prb(X\ge z)\,\mathrm dz.
\]
For a single threshold \(\alpha\in[0,1]\), splitting according to
whether \(X\ge\alpha\) gives the useful bounds
\[
 \alpha\Prb(X\ge\alpha)
 \le\E[X]
 \le\alpha+(1-\alpha)\Prb(X>\alpha).
\]
Consequently, a confidence guarantee valid for every input implies
\(R_F^{\mathcal A}(n)\ge(1-\delta)\alpha(n,\delta)\).
This is how the PROP1 probability bound in \cref{sec:prop1} yields
an expected guarantee. The upper inequality converts the small
success probabilities in \cref{sec:impossibility} into upper bounds 
on expected realized fairness.

\subsection{Relation to ex-ante and ex-post fairness}\label{subsec:exante-expost}

\paragraph{Ex-ante fairness evaluates expected utilities.}
For a fixed input \(I\), let \(A=A^{\mathcal A}(I)\) be the random output.
Exact ex-ante EF requires
\(\E[v_i(A_i)]\ge\E[v_i(A_j)]\) for every \(i,j\),
while exact ex-ante PROP requires \(\E[v_i(A_i)]\ge {V_i}/{n}\) for every \(i\).
The rule \(\Rand\) chooses each recipient independently and uniformly.
More generally, if each good has probability \({1}/{n}\) of going to each
agent, linearity of expectation gives
\[
 \E[v_i(A_j)]
 =\sum_{g\in M}v_i(g)\Prb(g\in A_j)
 =\frac{V_i}{n} \qquad(i,j\in[n]).
\]
This identity proves exact ex-ante EF and PROP without requiring
independence across goods.

\paragraph{Ex-post fairness evaluates every supported outcome.}
For \(\alpha\in[0,1]\), an algorithm has a support-wide ex-post
\(\alpha\)-\(F\) guarantee if, for every input, every allocation in its
support has score at least \(\alpha\). Since a fixed input has at most
\(n^m\) allocations, this is equivalent to requiring that, for every
fixed input \(I\),
\[
 \Prb_{\mathcal A}\bigl(\rho_F(\mathcal A,I)\ge\alpha\bigr)=1.
\]
It is the zero-failure endpoint of a confidence guarantee and implies
\(R_F^{\mathcal A}(n)\ge\alpha\).
Expected realized and high-probability fairness allow some outcomes
to have lower scores.
The following lemma shows that randomization does not improve
support-wide guarantees in this model.

\begin{lemma}[Randomization does not improve ex-post guarantees]\label[lemma]{lem:expost-equivalence}
For any \(\alpha\ge0\), if a randomized fully online algorithm \(\mathcal A\) has a support-wide
ex-post \(\alpha\)-\(F\) guarantee against a non-adaptive adversary,
then some deterministic fully online algorithm has the same guarantee.
\end{lemma}

\begin{proof}
Neoh et al.~\cite[Lemma~A.1]{NeohPetersTeh2026}
prove a related result against adaptive adversaries, requiring
the guarantee to hold for every realization of the algorithm's internal
randomness. Here, against a non-adaptive adversary, the guarantee is
required only for outputs with positive probability on each fixed
input. We therefore construct a deterministic rule whose output belongs
to the corresponding support.

At each arrival, choose the smallest-index recipient to whom
\(\mathcal A\) assigns the current good with positive conditional
probability, given the revealed valuation vectors and the previously
selected recipients. Since \(\mathcal A\) is online, this defines
a single deterministic online rule \(\mathcal D\) without future
information. For every fixed finite input \(I\), induction on the
arrivals shows that the recipient history selected by \(\mathcal D\)
has positive probability under \(\mathcal A\). Hence
\(A^{\mathcal D}(I)\) belongs to the support of \(A^{\mathcal A}(I)\),
and the support-wide guarantee gives
\(\rho_F(A^{\mathcal D}(I))\ge\alpha\).

Since every deterministic algorithm is also a degenerate randomized
algorithm, deterministic and randomized online algorithms have
the same achievable support-wide factors in this model.
\end{proof}

For every fixed \(n\ge2\), Neoh and
Teh~\cite[Theorem~3.1]{NeohTeh2026} rule out every
positive deterministic approximation to PROP\(k\) for each fixed
\(k\ge1\), and derive corresponding impossibilities for a broad range
of envy-based, proportionality-based, and share-based fairness notions,
including PROP1, EF1, and EFX. For a fixed deterministic algorithm,
their adaptive constructions produce fixed finite bad inputs.
\cref{lem:expost-equivalence} therefore rules out every positive
randomized support-wide ex-post factor for PROP1, EF1, and EFX against a
non-adaptive adversary.

These impossibilities alone do not rule out high-probability guarantees
or bound the expected realized score, because a bad supported outcome
may have arbitrarily small probability. The finite-distribution
arguments in \cref{sec:impossibility} control the relevant probabilities
on fixed inputs to obtain expected-score upper bounds.

%% file: sections/prop1_osq.tex
\section{Improved PROP1 Guarantees}\label{sec:prop1}

We introduce the \emph{Ordered Scale Quota} algorithm
(\(\OSQ\)), a randomized online allocation rule whose group size depends
only on \(n\). Against a non-adaptive adversary, for every fixed \(c>0\),
all sufficiently large \(n\), and every fixed input with \(n\) agents,
\(\OSQ\) achieves a realized PROP1 factor of at least
\(\frac{3-\sqrt5}{2}-o(1)\) with probability at least \(1-n^{-c}\).
The same asymptotic lower bound holds for its expected
realized PROP1 guarantee.

At the same confidence level \(1-\delta\), \(\OSQ\) improves the
high-probability PROP1 guarantee of \(\Rand\) by a factor
of \(\Omega(\log n)\), uniformly over \(\delta\in(0,1/2]\).
We also prove the sharp asymptotic guarantee
\((1+o(1))/\log n\) for \(\Rand\) under the expected realized
criterion. Thus, under this criterion,
\(\OSQ\) improves on \(\Rand\) by a factor of at least
\(\bigl(\frac{3-\sqrt5}{2}-o(1)\bigr)\log n\).

The proof combines a deterministic certificate based on the shared
scale records with an exponential-moment bound obtained by induction
along the fixed agent order.

\subsection{The Ordered Scale Quota Algorithm}

\paragraph{High-level idea.}
Under \(\Rand\), goods are allocated independently, so an agent may
miss many goods she values.  Our algorithm instead lets goods pass
through agents in a fixed order and applies a quota separately to each
agent and value scale: exactly one of every \(k\) consecutive visits
to a record is accepted, while the others continue through the
algorithm.  Thus allocations of similarly valued goods are coordinated
rather than independent.  Every completed group gives the agent one
good, and at most one group per scale remains unfinished.  This
structure allows the algorithm to use more opportunities to consider
valuable goods while still guaranteeing a constant realized PROP1
factor.

\paragraph{Formal description.}
Fix \(q=\frac{3+\sqrt5}{2}\), an integer \(k\ge2\), and the deterministic agent order
\(1,2,\ldots,n\).  For each agent \(j\in[n]\), define the activation
probability
\begin{equation}\label{eq:osq-pj}
    p_j:=\frac{k}{n+k-j}.
\end{equation}
Since \(n+k-j\ge k\), we have \(p_j\in(0,1]\).

For every agent \(j\) and value scale \(d\), \(\OSQ\) maintains a single
persistent record \((j,d)\).  Positive values use the scales
\([q^d,q^{d+1})\), \(d\in\mathbb Z\), while zero values use a separate scale
\(\bot\).  Consecutive \(k\) visits to the same record form one group.
At the first visit of a group, a winning position is drawn uniformly from
\([k]\).  The group closes only after its \(k\)-th visit, regardless of when
the winning position appears.

When a good \(g\) arrives, agents are scanned in the fixed order.
At agent \(j\), an independent Bernoulli activation with probability \(p_j\)
is drawn.  If \(j\) is not activated, the good is passed to the next agent
without visiting a record.  If \(j\) is activated, the good visits the record
determined by the scale of \(v_j(g)\).  A visit at the winning position is
accepted by \(j\); every other visit is forwarded.  If no agent accepts the
good, it is assigned by a fresh uniform fallback draw from \([n]\).

\begin{algorithm}[H]
\caption{Ordered Scale Quota (\(\OSQ(k)\))}
\label{alg:osq}
\begin{algorithmic}[1]
\Require \(n\ge2\) agents and an integer \(k\ge2\).
\Statex \textbf{State:} Record \((j,d)\) is indexed by agent \(j\) and scale \(d\).
\Statex \hspace{\algorithmicindent}
\(c_{jd}\) counts visits in the current group, and \(W_{jd}\) is its winning position.
\Statex \textbf{Randomness:} All activation, winning-position, and fallback draws are mutually independent.
\State Initialize an empty collection of records, retained across goods.
\For{each arriving good \(g\) with revealed values \(v(g)=(v_i(g))_{i\in[n]}\)}
    \For{\(j=1,\ldots,n\)}
        \State Draw \(B_j(g)\sim\operatorname{Bernoulli}(p_j)\), where
        \(p_j=\frac{k}{n+k-j}\).
        \If{\(B_j(g)=0\)}
            \State \textbf{continue} to agent \(j+1\).
        \EndIf
        \State Set \(d=\lfloor\log_q v_j(g)\rfloor\) if \(v_j(g)>0\), and
        \(d=\bot\) otherwise.
        \If{record \((j,d)\) does not exist}
            \State Create it with \(c_{jd}=0\).
        \EndIf
        \If{\(c_{jd}=0\)}
            \State Draw a fresh \(W_{jd}\) uniformly from \([k]\).
        \EndIf
        \State Set \(c_{jd}\gets c_{jd}+1\).
        \State Set \(\mathit{win}\gets(c_{jd}=W_{jd})\).
        \If{\(c_{jd}=k\)}
            \State Reset \(c_{jd}=0\).
        \EndIf
        \If{\(\mathit{win}\)}
            \State Assign \(g\) to \(j\) and skip to the next good.
        \EndIf
    \EndFor
    \State Assign \(g\) to a fresh uniform agent in \([n]\).
\EndFor
\end{algorithmic}
\end{algorithm}


The activation probabilities are chosen to balance the allocation
probabilities across the fixed agent order. The following proposition
shows that every good has a uniform marginal recipient, which also gives
exact ex-ante EF and PROP.

\begin{proposition}[Uniform marginals]\label[proposition]{lem:osq-marginals}
For every fixed input, every good \(g\), and every agent \(i\),
\[
    \Prb_{\OSQ(k)}(g\to i)=\frac1n.
\]
Consequently, \(\OSQ(k)\) satisfies ex-ante EF and ex-ante PROP.
\end{proposition}

\begin{proof}
Fix a good \(g\). Conditional on reaching filter \(j\), it activates
\(j\) with probability \(p_j\); conditional on activation, its visit is
the winning position with probability \(1/k\). Thus \(g\) is accepted
at filter \(j\) with probability \(p_j/k\), and otherwise forwarded.

Regard fallback as stage \(n+1\). A good reaches stage \(i\) exactly
when all preceding filters forward it. Hence, by \eqref{eq:osq-pj},
for every \(i\in[n+1]\), the probability that \(g\) reaches stage \(i\) is
\[
    \prod_{j=1}^{i-1}\left(1-\frac{p_j}{k}\right)=
    \prod_{j=1}^{i-1}\frac{n+k-j-1}{n+k-j}
    =
    \frac{n+k-i}{n+k-1}.
\]
In particular, the probability of reaching fallback is
\(\frac{k-1}{n+k-1}\). Since the fallback recipient is uniform, for
every \(i\in[n]\),
\[ \begin{aligned} \Prb(g\to i) &= \left(\prod_{j=1}^{i-1} \left(1-\frac{p_j}{k}\right)\right)\frac{p_i}{k} +\frac1n\prod_{j=1}^{n} \left(1-\frac{p_j}{k}\right)\\ &= \frac{n+k-i}{n+k-1}\cdot\frac1{n+k-i} +\frac1n\cdot\frac{k-1}{n+k-1} =\frac1n. \end{aligned} \]

Consequently, for every pair of agents \(i,\ell\), additivity and
linearity of expectation give
\[
    \E[v_i(A_\ell)]
    =
    \sum_g v_i(g)\Prb(g\to\ell)
    =
    \frac{V_i}{n}.
\]
Thus every agent has the same expected value for every bundle, and
\(\OSQ(k)\) satisfies ex-ante EF and ex-ante PROP.
\end{proof}

\subsection{The Main PROP1 Guarantee}

Recall that \(U_i=v_i(A_i)\) and
\(b_i=\max_{g\notin A_i}v_i(g)\), with \(b_i=0\) if \(A_i\) contains
all goods.  Thus \(U_i+b_i\) is the realized PROP1 quantity of agent \(i\).

We prove the main guarantee in two steps. First, we establish a
deterministic certificate relating the total value of visits to an
agent's records to her realized PROP1 quantity \(U_i+b_i\). We then
control the lower tail of this total through an exponential-moment
bound. We begin with the deterministic step.

Fix a deterministic input, a realization of the algorithmic randomness, and
an agent \(i\).  Let \(S_i\) be the total value to \(i\) of goods that reach
filter \(i\) and activate \(i\); equivalently, \(S_i\) is the total
\(i\)-value of all visits to records of agent \(i\).  Let \(X_i\) be the
value received by \(i\) through acceptance at its own filter, and let
\(U_i^{\mathrm{fb}}\) be the value received by \(i\) through fallback.
Because an item that passes filter \(i\) can subsequently be assigned only
to an agent \(j>i\) or by fallback,
\[
    U_i=X_i+U_i^{\mathrm{fb}}.
\]

\begin{lemma}[Deterministic certificate]\label[lemma]{lem:certificate}
For every realization,
\[
    S_i\le\bigl(1+q(k-1)\bigr)(U_i+b_i).
\]
\end{lemma}

\begin{proof}
Partition all groups of records \((i,d)\) according to whether their winning
position has appeared by the end of the input.

\emph{Groups whose winning position has appeared.}
Consider such a group and let \(w\) be the value of its winning good.
If \(w=0\), the group contributes zero to \(S_i\).  Otherwise, every observed
value in the group lies in the same scale interval as \(w\), and therefore
every other observed value is strictly smaller than \(qw\).  There are at
most \(k-1\) other observed positions, so the total value of the observed
positions in the group is at most
\(
    w+q(k-1)w=\bigl(1+q(k-1)\bigr)w.
\)
The winning good is assigned to \(i\).  This argument also covers an
unfinished current group whose winner has already appeared.  Summing over
all groups in this class gives a contribution at most
\[
\bigl(1+q(k-1)\bigr)X_i.
\]

\emph{Groups whose winning position has not appeared.}
Every such group must be the unique current unfinished group of its record:
a completed group contains all \(k\) positions and hence necessarily contains
its winner.  Therefore, for each value scale there is at most one such group,
it contains at most \(k-1\) observed positions, and every observed good in it
has been forwarded past filter \(i\).

Among these forwarded goods, those finally assigned to \(i\) can only return
through fallback, so their total \(i\)-value is at most
\(U_i^{\mathrm{fb}}\).

Now consider those forwarded goods that are finally outside \(A_i\).
If all have value zero, they contribute nothing.  Otherwise let \(z>0\) be
their largest \(i\)-value, and choose \(d\) with
\(q^d\le z<q^{d+1}\).  There are at most \(k-1\) such goods in scale \(d\),
each worth at most \(z\), and at most \(k-1\) in every lower scale \(r<d\),
each worth less than \(q^{r+1}\).  Hence their total value is at most
\begin{align*}
    (k-1)\left(z+\sum_{r=-\infty}^{d-1}q^{r+1}\right)
    &=(k-1)\left(z+\frac{q^{d+1}}{q-1}\right)\\
    &\le(k-1)\left(1+\frac{q}{q-1}\right)z\\
    &=q(k-1)z\\
    &\le q(k-1)b_i,
\end{align*}
where the equality in the penultimate line uses
\(q=2+1/(q-1)\).

Combining the two classes,
\[
\begin{aligned}
    S_i
    &\le
    \bigl(1+q(k-1)\bigr)X_i
    +U_i^{\mathrm{fb}}
    +q(k-1)b_i\\
    &\le
    \bigl(1+q(k-1)\bigr)U_i
    +q(k-1)b_i\\
    &\le
    \bigl(1+q(k-1)\bigr)(U_i+b_i).
\end{aligned}
\]
\end{proof}

The outside-good term \(b_i\) accounts for visits in unfinished groups
whose winning positions have not appeared and whose goods are ultimately
allocated elsewhere. There are at most \(k-1\) such visits per scale, so
the geometric sum over scales allows a single maximum outside good to
control their total value. It remains to control the lower tail of
\(S_i\). Since forwarding decisions within a group are dependent, we
first establish a product-moment inequality for one group and then apply
it successively through the filters.

\begin{lemma}[One-group product inequality]
\label[lemma]{lem:group}
Consider a group whose winning position \(W\) is uniform on \([k]\).
Suppose that its first \(r\) positions have appeared, where
\(0\le r\le k\), and let \(u_1,\ldots,u_r\in[0,1]\) be fixed. Then,
with the empty product interpreted as one,
\[
\E\!\left[
\prod_{\substack{1\le t\le r\\ t\,\textnormal{is forwarded}}} u_t
\right]
\le
\prod_{t=1}^r
\left(
\frac{1}{k}
+
\left(1-\frac{1}{k}\right)u_t
\right).
\]
\end{lemma}

\begin{proof}
The case \(r=0\) is immediate. For \(r\ge1\), averaging over the
winning position gives
\begin{equation}\label{eq:group-moment}
\E\!\left[
\prod_{\substack{1\le t\le r\\ t\,\textnormal{is forwarded}}} u_t
\right]
=
\frac{1}{k}\sum_{w=1}^r
\prod_{\substack{1\le t\le r\\ t\ne w}}u_t
+
\frac{k-r}{k}\prod_{t=1}^r u_t.
\end{equation}
Indeed, if \(W=w\le r\), every observed position except \(w\) is
forwarded, while if \(W>r\), all \(r\) observed positions are forwarded.

It remains to compare \eqref{eq:group-moment} with the claimed
product bound. The difference between the right-hand side of the
claimed inequality and the right-hand side of
\eqref{eq:group-moment} is affine in each \(u_t\) separately.
Therefore, fixing all other coordinates, its minimum over
\(u_t\in[0,1]\) is attained at \(u_t\in\{0,1\}\).
Applying this argument successively to all coordinates, its minimum
over \([0,1]^r\) is attained at a vertex. It therefore suffices to
verify the inequality for \(u_t\in\{0,1\}\) for every \(t\in[r]\).

Consider a vertex of \([0,1]^r\) with exactly \(z\) zero coordinates.
In the claimed inequality, the product bound on the right equals \(k^{-z}\). 
The expectation on
the left equals \(1\) if \(z=0\). If \(z=1\), it equals \(1/k\), since
the product is nonzero exactly when the unique zero position is the
winner. If \(z\ge2\), at least one zero position is forwarded regardless
of the winner, so the expectation is \(0\). Hence the desired inequality
holds at every vertex.
\end{proof}

The right-hand side of \cref{lem:group} is the product moment obtained
if the observed positions were forwarded independently with probability
\(1-1/k\). Thus, although the actual forwarding decisions within a group
are dependent, the lemma bounds their product moment by the corresponding
independent expression. This bound will allow us to work backwards from
a target agent through the preceding filters.

Fix a target agent \(i\), \(s>0\), and a normalization \(\nu>0\).
For \(a\in[0,1]\), define
\[
    F_i(a):=1-p_i+p_i e^{-sa},
\]
and recursively, for \(j=i-1,i-2,\ldots,1\),
\begin{equation}\label{eq:osq-Fi-recursion}
    F_j(a)
    :=
    \frac{p_j}{k}
    +
    \left(1-\frac{p_j}{k}\right)F_{j+1}(a).
\end{equation}
Clearly \(F_j(a)\in[0,1]\).

The base function \(F_i(a)\) accounts for activation at the target
filter: an activated good contributes its value to \(S_i\), regardless
of whether its visit wins. For \(j<i\), the constant term \(p_j/k\)
accounts for acceptance at filter \(j\), after which the good contributes
nothing to \(S_i\), while the term \((1-p_j/k)F_{j+1}(a)\) accounts for
forwarding. The next lemma shows that these functions give a product
bound for the exponential moment of the total value reaching and
activating the target filter.

\begin{lemma}[Exponential-moment bound]\label[lemma]{lem:osq-mgf}
Fix a target agent \(i\), a parameter \(s>0\), and a normalization
\(\nu>0\).  Let \(J\) be any deterministic incoming sequence satisfying
\(v_i(g)\le\nu\) for every \(g\in J\).  Consider a fresh copy of the filters
\(j,j+1,\ldots,i\), with all their records initially empty, and let
\(T_j(J)\) denote the total \(i\)-value of goods that reach filter \(i\) and
activate \(i\).  Then
\begin{equation}\label{eq:osq-product}
    \E\left[e^{-sT_j(J)/\nu}\right]
    \le
    \prod_{g\in J}
    F_j\left(\frac{v_i(g)}{\nu}\right).
\end{equation}
In particular, for the full input \(J=M\), we have \(S_i=T_1(J)\).
With \(\nu_i=\max_{g\in J} v_i(g)>0\), this yields
\[
    \E\!\left[e^{-sS_i/\nu_i}\right]
    \le
    \exp\!\left(
        -\frac{kV_i}{(n+k-1)\nu_i}(1-e^{-s})
    \right).
\]
\end{lemma}

\begin{proof}
We prove \eqref{eq:osq-product} by backward induction on \(j\).

For \(j=i\), each incoming good independently activates \(i\) with probability
\(p_i\).  Since \(T_i(J)\) counts an activated good regardless of whether its
visit subsequently wins, the winning positions of filter \(i\) are irrelevant.
Writing \(a_g=v_i(g)/\nu\),
\[
\begin{aligned}
    \E\left[e^{-sT_i(J)/\nu}\right]
    &=
    \prod_{g\in J}
    \left(1-p_i+p_i e^{-sa_g}\right)=
    \prod_{g\in J}F_i(a_g).
\end{aligned}
\]

Now let \(j<i\) and assume the claim for \(j+1\).

\emph{Fix the activations and winners, then apply induction.}
First condition on the complete activation array \(\{B_j(g)\}_{g\in J}\)
at filter \(j\). This fixes the visiting goods, their records, all group
boundaries, and all positions within those groups, without exposing the
winning positions. Let \(\mathcal W\) denote the collection of
pre-sampled winning positions of all groups visited at filter \(j\).
Conditioning next on \(\mathcal W\) fixes the forwarded subsequence
\(J'\subseteq J\), in its original order, which still satisfies
\(v_i(g)\le\nu\) for every \(g\in J'\).
Since \(j<i\), we have \(T_j(J)=T_{j+1}(J')\).  The filters
\(j+1,\ldots,i\) have fresh independent randomness and initially empty
records, so the induction hypothesis gives
\[
    \E\left[
        e^{-sT_j(J)/\nu}
        \,\middle|\,
        \{B_j(g)\}_{g\in J},\mathcal W
    \right]
    \le
    \prod_{g\in J'}F_{j+1}(a_g).
\]

\emph{Average the winners and activations.}
Keep the activation array fixed. Every non-activated good is forwarded, while the activated goods are partitioned into fixed visited groups. Applying \cref{lem:group} to each group with
\(u_g=F_{j+1}(a_g)\in[0,1]\), including any unfinished group, and using the independence of their winning positions, we obtain
\[
\E\!\left[
    \prod_{g\in J'}F_{j+1}(a_g)
    \,\middle|\,
    \{B_j(g)\}_{g\in J}
\right]
\le
\prod_{\substack{g\in J\\B_j(g)=0}}F_{j+1}(a_g)
\prod_{\substack{g\in J\\B_j(g)=1}}
\left[
    \frac1k+
    \left(1-\frac1k\right)F_{j+1}(a_g)
\right].
\]
This bound factorizes over the incoming goods and no longer depends on
their record or group boundaries. We now average the conditional inequality
obtained from the induction hypothesis first over \(\mathcal W\), keeping
the activation array fixed and applying the bound above, and then over
the activation array. By the tower property and the independence of the
Bernoulli activations, we obtain
\[
\begin{aligned}
\E\left[e^{-sT_j(J)/\nu}\right]
&\le
\E_{B_j}\!\left[
\prod_{\substack{g\in J\\B_j(g)=0}}F_{j+1}(a_g)
\prod_{\substack{g\in J\\B_j(g)=1}}
\left(
    \frac1k+
    \left(1-\frac1k\right)F_{j+1}(a_g)
\right)
\right]\\
&=
\prod_{g\in J}
\left[
    (1-p_j)F_{j+1}(a_g)
    +p_j\left(
        \frac1k+
        \left(1-\frac1k\right)F_{j+1}(a_g)
    \right)
\right]\\
&=
\prod_{g\in J}
\left[
    \frac{p_j}{k}
    +\left(1-\frac{p_j}{k}\right)F_{j+1}(a_g)
\right]\\
&=
\prod_{g\in J}F_j(a_g),
\end{aligned}
\]
where the last equality follows from \eqref{eq:osq-Fi-recursion}. This proves \eqref{eq:osq-product}.

It remains to solve the recursion.  By the definition of
$F_i(a)$ and recurrence relation, we have
\[
    1-F_1(a)
    =
    p_i
    \left(
        \prod_{j=1}^{i-1}\left(1-\frac{p_j}{k}\right)
    \right)
    (1-e^{-sa}).
\]
Therefore
\[
    F_1(a)
    =
    1-\frac{k}{n+k-1}(1-e^{-sa}).
\]

For the full input \(M\), \(T_1(M)=S_i\). Applying
\eqref{eq:osq-product} with \(\nu=\nu_i\), we obtain
\[
\begin{aligned}
    \E\!\left[e^{-sS_i/\nu_i}\right]
    &\le
    \prod_{g\in M}
    \left[
        1-\frac{k}{n+k-1}
        \left(1-e^{-s v_i(g)/\nu_i}\right)
    \right]\\
    &\le
    \exp\!\left(
        -\frac{k}{n+k-1}
        \sum_{g\in M}
        \left(1-e^{-s v_i(g)/\nu_i}\right)
    \right)\\
    &\le
    \exp\!\left(
        -\frac{k}{n+k-1}(1-e^{-s})
        \sum_{g\in M}\frac{v_i(g)}{\nu_i}
    \right)\\
    &=
    \exp\!\left(
        -\frac{kV_i}{(n+k-1)\nu_i}(1-e^{-s})
    \right).
\end{aligned}
\]
Here the second inequality uses \(1-x\le e^{-x}\), and the third uses
\(1-e^{-sa}\ge a(1-e^{-s})\) for \(a\in[0,1]\), which follows from
the convexity of \(e^{-sa}\).
\end{proof}

We now combine the deterministic and probabilistic estimates.
\Cref{lem:certificate} converts a lower bound on \(S_i\) into a lower
bound on \(U_i+b_i\), while \cref{lem:osq-mgf} controls the lower tail
of \(S_i\). Applying Markov's inequality and then a union bound over the
agents yields a simultaneous PROP1 guarantee. The group size in the
following theorem depends only on \(n\), so the same algorithm provides
the stated bounds for every failure probability \(\delta\).

\begin{theorem}[High-probability PROP1 guarantee]
\label[theorem]{thm:osq-high-probability}
For all sufficiently large \(n\), set
\[
    k_n:=\max\{2,\lceil(\log n)^2\rceil\},\qquad
    s_n:=(\log n)^{-1/2},\qquad
    \eta_n:=
    \frac{n}{n+k_n-1}\,
    \frac{1-e^{-s_n}}{s_n}.
\]
For every fixed input and
every \(\delta\in(0,1)\),
\begin{equation}\label{eq:osq-uniform-hp}
    \Prb\left(
        \rho_{\mathrm{PROP1}}
        \ge
        \frac{\eta_n}
        {q+\tfrac{\log(n/\delta)}{(\log n)^{3/2}}}
    \right)
    \ge 1-\delta.
\end{equation}
\end{theorem}

\begin{proof}
Write
\[
    L:=\log(n/\delta),\qquad
    C_{k_n}:=1+q(k_n-1),
\]
and define
\begin{equation}\label{eq:osq-alpha}
    \alpha_{n,\delta}
    :=
    \frac{k_ns_n\eta_n}
         {C_{k_n}s_n+L}.
\end{equation}
We first show that
\(\Prb(\rho_{\mathrm{PROP1}}\ge\alpha_{n,\delta})\ge1-\delta\).

Fix an agent \(i\). If \(V_i=0\), her score is one. Otherwise, let
\(\nu_i=\max_g v_i(g)>0\). Since \(U_i+b_i\ge\nu_i\), a violation
is impossible when \(\nu_i\ge\alpha_{n,\delta}V_i/n\).
In the remaining case, \(\nu_i<\alpha_{n,\delta}V_i/n\). Since
\(k_ns_n\eta_n=\frac{nk_n}{n+k_n-1}(1-e^{-s_n})\),
\cref{lem:certificate}, Markov's inequality, and
\cref{lem:osq-mgf} give
\[
\begin{aligned}
    \Prb\left(
        U_i+b_i<\frac{\alpha_{n,\delta}V_i}{n}
    \right)
    &\le
    \Prb\left(
        S_i<\frac{C_{k_n}\alpha_{n,\delta}V_i}{n}
    \right)\\
    &=
    \Prb\left(
        e^{-s_nS_i/\nu_i}
        >\exp\left(
            -\frac{C_{k_n}s_n\alpha_{n,\delta}V_i}{n\nu_i}
        \right)
    \right)\\
    &\le
    \exp\left(
        \frac{C_{k_n}s_n\alpha_{n,\delta}V_i}{n\nu_i}
    \right)
    \E\!\left[e^{-s_nS_i/\nu_i}\right]\\
    &\le
    \exp\left(
        -\frac{V_i}{n\nu_i}
        \left(
            k_ns_n\eta_n
            -C_{k_n}s_n\alpha_{n,\delta}
        \right)
    \right)\\
    &=
    \exp\left(
        -\frac{\alpha_{n,\delta}V_i}{n\nu_i}L
    \right)
    \le e^{-L}
       =\frac{\delta}{n}.
\end{aligned}
\]
Here \eqref{eq:osq-alpha} gives
\(k_ns_n\eta_n-C_{k_n}s_n\alpha_{n,\delta}=\alpha_{n,\delta}L\),
and the last inequality uses \(\alpha_{n,\delta}V_i/(n\nu_i)>1\).
A union bound over the agents with positive total value gives
\[
    \Prb(\rho_{\mathrm{PROP1}}<\alpha_{n,\delta})
    \le
    \sum_{i:V_i>0}
    \Prb\left(U_i+b_i<\frac{\alpha_{n,\delta}V_i}{n}\right)
    \le
    \sum_{i:V_i>0}\frac{\delta}{n}
    \le\delta.
\]
This proves
\(\Prb(\rho_{\mathrm{PROP1}}\ge\alpha_{n,\delta})\ge1-\delta\).

Finally, \(C_{k_n}/k_n\le q\) and
\(k_ns_n\ge(\log n)^{3/2}\). Dividing the numerator and denominator
of \eqref{eq:osq-alpha} by \(k_ns_n\) therefore yields
\[
\begin{aligned}
    \alpha_{n,\delta}
    &=
    \frac{
        \eta_n
    }{
        \tfrac{C_{k_n}}{k_n}
        +\tfrac{L}{k_ns_n}
    }\ge
    \frac{\eta_n}{q+L/(\log n)^{3/2}}.
\end{aligned}
\]
This proves \eqref{eq:osq-uniform-hp}, uniformly over
\(\delta\) and the input.
\end{proof}

We now compare the high-probability guarantees of \(\OSQ\) and
\(\Rand\) at the same confidence level. For the chosen parameters,
\(k_n=O((\log n)^2)\) and \(s_n=(\log n)^{-1/2}\), so
\[
    \eta_n
    =
    \left(1-O\!\left(\frac{(\log n)^2}{n}\right)\right)
    \left(1-O\!\left(\frac1{\sqrt{\log n}}\right)\right)
    =1-O\!\left(\frac1{\sqrt{\log n}}\right).
\]
Hence \cref{thm:osq-high-probability} gives the high-probability factor
\[
    \Omega\!\left(
        \frac{1}
        {1+\tfrac{\log(n/\delta)}{(\log n)^{3/2}}}
    \right)
\]
with probability at least \(1-\delta\), for every
\(\delta\in(0,1)\) and every fixed input.
Prior work shows that, for every \(\delta\in(0,1/2]\), the largest
factor that \(\Rand\) guarantees with probability at least
\(1-\delta\) on every fixed input is
\(\Theta(1/\log(n/\delta))\)
\cite{ChooEtAl2026,NeohTeh2026}.
\(\OSQ\) therefore improves this guarantee by a factor of at least
\[
    \Omega\!\left(
        \frac{\log(n/\delta)}
             {1+\tfrac{\log(n/\delta)}{(\log n)^{3/2}}}
    \right)
    =
    \Omega\!\left(\min\{\log(n/\delta),(\log n)^{3/2}\}\right).
\]
Since \(\log(n/\delta)\ge\log n\), this gives an improvement of
\(\Omega(\log n)\) uniformly over \(\delta\in(0,1/2]\).

For polynomially small failure probabilities, the confidence term in
the denominator of \eqref{eq:osq-uniform-hp} tends to zero, while
\(\eta_n\) tends to one. The
following corollary makes the resulting asymptotic factor \(1/q\)
explicit and derives an expected realized guarantee from the same tail
bound.

\begin{corollary}[Constant high-probability and expected guarantees]
\label[corollary]{cor:osq-expected}
Consider \(\OSQ\) with \(k=k_n\) as in
\cref{thm:osq-high-probability}. For every fixed constant \(c>0\),
all sufficiently large \(n\), and every fixed input,
\[
    \Prb\left(
        \rho_{\mathrm{PROP1}}
        \ge
        \frac{3-\sqrt5}{2}
        -O_c\!\left(\frac1{\sqrt{\log n}}\right)
    \right)
    \ge1-n^{-c}.
\]
Moreover,
\[
    R_{\mathrm{PROP1}}^{\OSQ}(n)
    \ge
    \frac{3-\sqrt5}{2}
    -O\!\left(\frac1{\sqrt{\log n}}\right).
\]
\end{corollary}

\begin{proof}
Set \(\delta=n^{-c}\) in
\cref{thm:osq-high-probability}. Then
\(\log(n/\delta)=(1+c)\log n\), and the threshold in
\eqref{eq:osq-uniform-hp} is
\[
    \frac{1-O((\log n)^{-1/2})}
         {q+(1+c)/\sqrt{\log n}}
    =
    \frac1q-O_c\!\left(\frac1{\sqrt{\log n}}\right).
\]
Since \(1/q=(3-\sqrt5)/2\), this proves the probability
claim. For the expected guarantee, take \(c=1\). The probability
bound is uniform over inputs, so nonnegativity of the realized score gives
\[
\begin{aligned}
    R_{\mathrm{PROP1}}^{\OSQ}(n)
    &=\inf_I\E\!\left[\rho_{\mathrm{PROP1}}(\OSQ,I)\right]\\
    &\ge
    \left(1-\frac1n\right)
    \left(
        \frac{3-\sqrt5}{2}
        -O\!\left(\frac1{\sqrt{\log n}}\right)
    \right)\\
    &=
    \frac{3-\sqrt5}{2}
    -O\!\left(\frac1{\sqrt{\log n}}\right).
\end{aligned}
\]
\end{proof}

\subsection{Expected Realized PROP1 for \(\Rand\)}

We next compare the algorithms under the expected realized criterion.
Existing results for \(\Rand\) imply an expected realized
PROP1 guarantee of order \(\Theta(1/\log n)\)
\cite{ChooEtAl2026,NeohTeh2026}. The following proposition determines
its leading constant, giving a sharp benchmark for comparison with the
constant guarantee of \(\OSQ\).

\begin{proposition}\label[proposition]{prop:rand-prop1}
As \(n\to\infty\), the expected realized \emph{PROP1} guarantee of \(\Rand\) satisfies
\[
 R_{\mathrm{PROP1}}^{\mathrm{Rand}}(n)=\frac{1+o(1)}{\log n}.
\]
\end{proposition}
\begin{proof}

\emph{Lower bound.}
For \(s>0\) and \(\delta\in(0,1)\), set
\[
 \beta=\frac{1-e^{-s}}{s+\log(n/\delta)}.
\]
Since \(0<1-e^{-s}<s\), we have \(0<\beta<1\).
Fix an input and a positive-total agent \(i\), and write
\(\nu_i=\max_g v_i(g)>0\).
Under \(\Rand\), the ownership indicators for agent \(i\) are independent
Bernoulli variables with parameter \({1}/{n}\). Thus
\[
 \E\!\left[e^{-sU_i/\nu_i}\right]
 =\prod_g\left(1-\frac{1}{n}+\frac{e^{-sv_i(g)/\nu_i}}{n}\right)
 \le \exp\!\left(-\frac{(1-e^{-s})V_i}{n\nu_i}\right).
\]
To justify the inequality, convexity gives
\(e^{-sa}\le1-(1-e^{-s})a\) for \(a\in[0,1]\).
Apply this to \(a={v_i(g)}/{\nu_i}\), then use \(1-x\le e^{-x}\)
on each factor and sum the values over goods.

As before, \(U_i+b_i\ge\nu_i\), so a violation of
\(\beta\)-PROP1 is impossible when \(\nu_i\ge{\beta V_i}/{n}\).
Otherwise that violation implies \(U_i<{\beta V_i}/{n}\).
Passing to the decreasing exponential and applying Markov's inequality gives
\begin{align*}
 \Prb\left(U_i+b_i<\frac{\beta V_i}{n}\right)
 &\le \Prb\left(e^{-sU_i/\nu_i}
       >\exp\!\left(-\frac{s\beta V_i}{n\nu_i}\right)\right)\\
 &\le \exp\!\left(\frac{s\beta V_i}{n\nu_i}\right)
       \E\!\left[e^{-sU_i/\nu_i}\right]\\
 &\le \exp\!\left(-\frac{V_i}{n\nu_i}(1-e^{-s}-s\beta)\right)\\
 &=\exp\!\left(-\frac{\beta V_i}{n\nu_i}\log\left(\frac{n}{\delta}\right)\right)
 \le \frac{\delta}{n}.
\end{align*}
The equality uses the definition of \(\beta\), and the last inequality
uses \(\frac{\beta V_i}{n\nu_i}>1\).
Zero-total agents automatically have score one.
A union bound gives
\(\Prb(\rho_{\mathrm{PROP1}}\ge\beta)\ge1-\delta\), hence
\[
 R_{\mathrm{PROP1}}^{\mathrm{Rand}}(n)
 \ge(1-\delta)\frac{1-e^{-s}}{s+\log(n/\delta)}.
\]
For sufficiently large \(n\), choose \(s=\log\log n\) and
\(\delta={1}/{\log n}\). Then
\[
 R_{\mathrm{PROP1}}^{\mathrm{Rand}}(n)
 \ge\frac{(1-1/\log n)^2}{\log n+2\log\log n}
 =\frac{1-o(1)}{\log n}.
\]

\emph{Upper bound.}
For sufficiently large \(n\), let
\(\ell=\lfloor\log n-2\log\log n\rfloor\ge1\).
Create disjoint blocks of \(n\ell\) goods, one block per agent.
Agent \(i\) values every good in her own block at one and all other
goods at zero. Fix any arrival order of these blocks.
Then \({V_i}/{n}=\ell\).
Her own valued-good count \(X_i\) under \(\Rand\) follows a binomial distribution, which is
\(X_i\sim\operatorname{Bin}\left(n\ell,{1}/{n}\right)\).
The \(X_i\)'s are independent because they depend on recipient draws
for disjoint sets of goods. For \(n\ge2\),
\[
 -\log\left(1-\frac{1}{n}\right)=\int_0^{\frac{1}{n}}\frac{dt}{1-t}
 \le\frac{1}{n-1}.
\]
Consequently, writing \(p_0:=\Prb(X_i=0)=\left(1-{1}/{n}\right)^{n\ell}\), we obtain
\begin{align*}
 np_0
 &=\exp\!\left(\log n+n\ell\log\!\left(1-\frac{1}{n}\right)\right)\\
 &\ge\exp\!\left(\log n-\frac{n\ell}{n-1}\right)\\
 &\ge\exp\!\left(\log n-\frac{n}{n-1}
                      \bigl(\log n-2\log\log n\bigr)\right)\\
 &\ge\exp\!\left(2\log\log n-\frac{\log n}{n-1}\right)\\
 &=(\log n)^2\cdot\exp\!\left(-\frac{\log n}{n-1}\right)
 \ge\frac{(\log n)^2}{2}.
\end{align*}
The first inequality uses the logarithmic bound above, and the second uses
\(\ell\le\log n-2\log\log n\). The third discards the nonnegative term
\(\frac{2\log\log n}{n-1}\) in the exponent, for sufficiently large \(n\).
The final inequality uses \(\frac{\log n}{n-1}\le\log2\),
which follows from \(n\le2^{n-1}\) for every integer \(n\ge2\).
If some \(X_i=0\), then \(U_i=0\) and \(b_i=1\), so that agent's
PROP1 score is \({1}/{\ell}\), and the simultaneous score is at most
\({1}/{\ell}\). Independence gives
\[
 \Prb(X_i>0\text{ for all }i)=(1-p_0)^n
 \le e^{-np_0}\le \exp\!\left(-\frac{(\log n)^2}{2}\right).
\]
Bounding the score by one on this latter event yields
\[
 \E[\rho_{\mathrm{PROP1}}]\le\frac{1}{\ell}+\exp\!\left(-\frac{(\log n)^2}{2}\right)
 =\frac{1+o(1)}{\log n}.
\]
This fixed-input family bounds the infimum from above and completes
the proof.
\end{proof}

By \cref{cor:osq-expected},
the expected realized PROP1 guarantee of \(\OSQ\) is at least
\(\frac{3-\sqrt5}{2}-o(1)\). \Cref{prop:rand-prop1} shows that the
corresponding guarantee of \(\Rand\) is \((1+o(1))/\log n\). Hence
\(\OSQ\) improves the expected realized PROP1 guarantee
of \(\Rand\) by a factor of at least
\(\bigl(\frac{3-\sqrt5}{2}-o(1)\bigr)\log n\).

%% file: sections/impossibility.tex
\section{Limitations for Envy-Based Fairness}\label{sec:impossibility}

We establish limitations on high-probability and expected realized
envy-based fairness for every randomized algorithm in our model.
For every such algorithm and every positive EFX factor, the success
probability can be made arbitrarily small on a suitable fixed input.
Consequently, the optimal expected realized EFX guarantee is zero.
For EF1, no such algorithm can guarantee a positive factor with
confidence exceeding \(\frac{n+1}{2n}\). The optimal expected realized
EF1 guarantee is also at most \(\frac{n+1}{2n}\).

The proofs use finite families of fixed inputs for a non-adaptive adversary.
For a fixed target factor \(\alpha>0\), we bound the average success
of every deterministic rule over such a family, and then average this
bound over the random tape of an arbitrary randomized algorithm,
following the standard argument underlying Yao's
principle~\cite{Yao1977}.
A fixed input with small success probability is then selected from
the family. The expected bounds follow because every realized score
lies in \([0,1]\).

\subsection{EFX Impossibility}


\begin{theorem}[EFX success probabilities]\label{thm:efx-success}
For every \(n\ge2\), every randomized fully online algorithm
\(\mathcal A\), every \(\alpha\in(0,1]\), and every \(p\in(0,1)\),
there exists a fixed finite input \(I\), independent of the algorithm's random choices,
such that
\[
 \Prb_{\mathcal A}(\rho_{\mathrm{EFX}}(\mathcal A,I)\ge\alpha)<p.
\]
Moreover, \(I\) can be chosen to have identical valuations across all agents.
\end{theorem}
\begin{proof}
Fix a randomized algorithm \(\mathcal A\),
\(\alpha\in(0,1]\), and \(p\in(0,1)\). Choose
\(0<\varepsilon<\alpha\le1\) and an integer \(K>\frac{n}{p}\).
Choose common values \(s_1,\ldots,s_K>0\), starting with \(s_1=1\)
and satisfying
\[
 s_t\ge\varepsilon^{-1}\sum_{h=1}^{t-1}s_h
 \quad (2\le t\le K).
\]
Before running the algorithm, choose \(T\) uniformly from \([K]\)
and present \(I_T=(s_1,\ldots,s_T)\), with every agent assigning
value \(s_t\) to good \(t\). These \(K\) fixed finite inputs form the family
\(\mathcal I_K=\{I_1,\dots,I_K\}\).

First consider any deterministic rule \(D\).
Its decisions on each \(I_T\) agree with the first \(T\) decisions
of its run on \(I_K\), because the horizon is not announced and
the revealed prefixes are identical.
Call time \(t\) fresh if its recipient was empty just before receiving
good \(t\). There are at most \(n\) fresh times on this common run
because each one changes one empty bundle into a nonempty bundle, and
irrevocability prevents that bundle from becoming empty again.

At a nonfresh time \(t\), let \(j\) receive the current good and put
\(P_{t-1}=\sum_{h<t}s_h\). Its bundle contains at least one older good.
Since \(\varepsilon<1\), we have \(s_t>P_{t-1}\), so the current
good is strictly larger than every older good. In particular,
a least-valued good in \(A_j^t\) is older.
Removing such a good leaves the current good, and therefore
\[
 d_{ij}^{\mathrm{EFX}}(A^t)\ge s_t\qquad(i\ne j).
\]
Every other agent owns only older goods, so
\(v_i(A_i^t)\le P_{t-1}\). There is at least one such agent because
\(n\ge2\). Its directed comparison with \(j\) gives
\[
 \rho_{\mathrm{EFX}}(A^t)
 \le \frac{P_{t-1}}{s_t}\le\varepsilon.
\]
Since \(\varepsilon<\alpha\), only fresh times can have score at
least \(\alpha\). There are at most \(n\) such times, so
\[
 \frac{1}{K}\sum_{T=1}^K
 \ind\{\rho_{\mathrm{EFX}}(D,I_T)\ge\alpha\}
 \le\frac{n}{K}.
\]
This counts successes on the same fixed list of inputs for every
deterministic rule.

Now fix the entire random tape of \(\mathcal A\), obtaining a
deterministic rule, and apply this inequality. Averaging over tapes
and interchanging the finite sum with expectation gives
\[
 \frac{1}{K}\sum_{T=1}^K
 \Prb_{\mathcal A}(\rho_{\mathrm{EFX}}(\mathcal A,I_T)\ge\alpha)
 \le\frac{n}{K}<p.
\]
Consequently, at least one fixed input in \(\mathcal I_K\) has
success probability at most this average. This input is selected
without observing the realized random tape, which proves the
theorem against a non-adaptive adversary.
\end{proof}

The bad fixed prefix is selected after averaging over the random
tape, not by observing one execution and choosing when to stop it.
All values in the construction are strictly positive, so the proof
also applies to EFX conventions that require the condition only for
positively valued goods.

\begin{corollary}[Zero expected EFX]
\label[corollary]{cor:efx-expected}
For every \(n\ge2\) and every randomized fully online algorithm
\(\mathcal A\),
\[
R_{\mathrm{EFX}}^{\mathcal A}(n)=0,
\]
even when the infimum is restricted to inputs with identical valuations.
\end{corollary}
\begin{proof}
Fix a randomized algorithm \(\mathcal A\) and
\(\alpha,p\in(0,1)\). On the input supplied by
\cref{thm:efx-success}, write \(\rho=\rho_{\mathrm{EFX}}(\mathcal A,I)\).
Because \(0\le\rho\le1\),
\[
 \E[\rho]
 \le\alpha\Prb(\rho<\alpha)+\Prb(\rho\ge\alpha)
 \le\alpha+p.
\]
Consequently \(R_{\mathrm{EFX}}^{\mathcal A}(n)\le\alpha+p\).
Letting \(\alpha,p\rightarrow 0\) proves the corollary.
\end{proof}

\subsection{An EF1 Upper Bound}

For every randomized online algorithm and every positive EF1
factor, we construct a fixed input on which the success probability is
at most \(\frac{n+1}{2n}\). The expected bound will follow by taking
the factor arbitrarily small.

\begin{theorem}[EF1 success probabilities]\label{thm:ef1-success}
For every \(n\ge2\), every randomized fully online algorithm
\(\mathcal A\), and every \(\alpha\in(0,1]\), there exists a fixed
finite input \(I\), independent of the algorithm's random choices,
such that
\[
 \Prb_{\mathcal A}(\rho_{\mathrm{EF1}}(\mathcal A,I)\ge\alpha)
 \le\frac{n+1}{2n}.
\]
\end{theorem}

\begin{proof}
Fix \(\mathcal A\) and \(\alpha\in(0,1]\), and choose a real
\(K>\frac{1}{\alpha}\ge1\), so that \(\frac{1}{K}<\alpha\).
Every input in our distribution begins with \(n-1\) goods
\(g_1,\ldots,g_{n-1}\), each worth one to every agent.
Before the run, choose a target agent \(a\) uniformly from \([n]\)
and independently choose a short or long horizon with equal
probability.

The next good \(g_n\) has value \(K\) to \(a\) and
\(\frac{1}{K}\) to every other agent. The short input is
$
I_a^{\mathrm S}=(g_1,\ldots,g_n).
$
The long input adds one more good \(g_{n+1}\), worth \(K\) to \(a\)
and one to every other agent, and is therefore
$
I_a^{\mathrm L}=(g_1,\ldots,g_{n+1}).
$
Thus, the distribution consists of \(2n\) equally likely fixed
inputs, and every good has strictly positive value to every agent.

Fix a deterministic rule \(D\). Its allocation of the common unit
prefix \(g_1,\ldots,g_{n-1}\) does not depend on either hidden
choice. For each target \(a\), its decision on \(g_n\) is also
identical on the short and long inputs, since those inputs are
indistinguishable until after \(g_n\) is allocated.
We distinguish two cases according to whether the first \(n-1\) goods
have distinct recipients.

\emph{Repeated recipients for the first \(n-1\) goods.}
Suppose two of the first \(n-1\) goods share a recipient.
Then at most \(n-2\) agents are occupied, so at least two are empty,
and some occupied bundle contains at least two unit goods.
After \(g_n\) is allocated, at least one agent \(e\) is still empty.
That occupied bundle still contains its two unit goods.
Deleting any one good leaves positive value to \(e\).
Since \(e\)'s own value is zero, the short-input score
is zero for every target. Only the \(n\) long inputs can possibly
have score at least \(\alpha\), so the success fraction is at most
\[
\frac{1}{2}\le\frac{n+1}{2n}.
\]

\emph{Distinct recipients for the first \(n-1\) goods.}
Otherwise, each unit good has a different recipient.
There is exactly one empty agent \(b\), and every other agent owns
one unit good. When \(a=b\), allow both horizons to succeed.
Now fix \(a\ne b\), so agent \(a\) initially has own value one.
There are two possibilities for the allocation of \(g_n\).

If \(g_n\) does not go to \(b\), then \(b\) remains empty, and the
recipient of \(g_n\) owns both its unit good and \(g_n\).
Both goods are strictly positive to \(b\), hence 
the short-input score is zero, so at most the long input can succeed.

If \(g_n\) goes to \(b\), allow the short input to succeed.
We claim that the long-input score is at most
\(\frac{1}{K}<\alpha\), regardless of how the algorithm allocates
\(g_{n+1}\).
If \(g_{n+1}\) also goes to \(b\), agent \(a\) keeps her unit good.
She values \(b\)'s two goods \(g_n\) and \(g_{n+1}\) at \(K\) and
\(K\), respectively. After deleting one of them, the EF1 denominator
remains \(K\), while her own utility is one. Hence the corresponding
EF1 factor is $\frac{1}{K}$.
If \(g_{n+1}\) goes to an agent \(j\ne b\), then \(j\) already owns
a unit good. Agent \(b\) is not the target and values her own good
\(g_n\) at \(\frac{1}{K}\), while she values both \(j\)'s unit good
and \(g_{n+1}\) at one. Her comparison with \(j\) therefore has
EF1 denominator one and factor $\frac{1}{K}$.
This argument also covers the case \(j=a\).
These two cases exhaust all possible recipients of \(g_{n+1}\).

Thus, for each of the \(n-1\) targets \(a\ne b\), at most one of
the corresponding short and long inputs succeeds. For \(a=b\), at
most two inputs succeed. There are therefore at most $2+(n-1)=n+1$
successful inputs. The empty agent \(b\) is determined by the common
unit prefix \(g_1,\ldots,g_{n-1}\), before the target is revealed.
Writing
\(
\mathcal I_K
=
\{I_a^{\mathrm S},I_a^{\mathrm L}:a\in[n]\}
\)
for the collection of \(2n\) inputs, we obtain
\[
\frac{1}{2n}
\sum_{I\in\mathcal I_K}
\ind\{\rho_{\mathrm{EF1}}(D,I)\ge\alpha\}
\le
\frac{n+1}{2n}.
\]
The first case satisfies the same bound.

Fix each random tape of \(\mathcal A\), apply this deterministic
inequality, and average over the random tapes. Since the input
distribution is independent of these tapes,
\[
\frac{1}{2n}
\sum_{I\in\mathcal I_K}
\Prb_{\mathcal A}
\bigl(\rho_{\mathrm{EF1}}(\mathcal A,I)\ge\alpha\bigr)
\le
\frac{n+1}{2n}.
\]
Consequently, at least one fixed input in \(\mathcal I_K\) has
success probability at most this average. This input is selected
without observing the realized random tape, which proves the
theorem against a non-adaptive adversary.
\end{proof}

\begin{corollary}[Expected EF1 upper bound]
\label[corollary]{cor:ef1-expected}
For every \(n\ge2\) and every randomized fully online algorithm
\(\mathcal A\),
\[
R_{\mathrm{EF1}}^{\mathcal A}(n)\le\frac{n+1}{2n}.
\]
\end{corollary}

\begin{proof}
Fix \(\mathcal A\) and \(\alpha\in(0,1)\). On the fixed input
\(I_\alpha\) given by \cref{thm:ef1-success}, write
\(
\rho=\rho_{\mathrm{EF1}}(\mathcal A,I_\alpha).
\)
Since \(0\le\rho\le1\),
\begin{align*}
\E[\rho]
&\le
\alpha\Prb(\rho<\alpha)+\Prb(\rho\ge\alpha)\\
&=
\alpha+(1-\alpha)\Prb(\rho\ge\alpha)\\
&\le
\frac{n+1}{2n}+\frac{n-1}{2n}\alpha.
\end{align*}
Consequently \(R_{\mathrm{EF1}}^{\mathcal A}(n)\le\frac{n+1}{2n}+\frac{n-1}{2n}\alpha\).
Letting \(\alpha\rightarrow 0\) proves the corollary.
\end{proof}

%% file: sections/conclusion.tex
\section{Conclusion}

We study high-probability fairness and expected realized
fairness as two complementary criteria for randomized online allocation.
They evaluate the same realized fairness factor through its lower tail
and its expectation. For PROP1, \(\OSQ\) uses only the number of agents as
advance information and preserves exact ex-ante envy-freeness and
proportionality.

Against a non-adaptive adversary, \(\OSQ\) improves the
high-probability PROP1 guarantee of independent uniform allocation
\(\Rand\) by a factor of \(\Omega(\log n)\), uniformly over
\(\delta\in(0,1/2]\) at the same confidence level \(1-\delta\).
For every fixed \(c>0\), all sufficiently large \(n\), and every fixed
input with \(n\) agents, it achieves a
\(\frac{3-\sqrt5}{2}-o(1)\) factor with probability at least
\(1-n^{-c}\).
The expected realized PROP1 guarantee of \(\OSQ\) is at least
\(\frac{3-\sqrt5}{2}-o(1)\). We also establish the asymptotically sharp
expected realized guarantee \((1+o(1))/\log n\) for \(\Rand\).
Thus \(\OSQ\) improves on \(\Rand\) under this criterion by a factor
of at least \(\bigl(\frac{3-\sqrt5}{2}-o(1)\bigr)\log n\).

For EFX and EF1, we establish limitations that apply to every randomized
fully online algorithm. For every such algorithm, the success probability
of any positive EFX factor can be made arbitrarily small, and its
expected realized EFX guarantee is zero. For EF1, no positive factor
can be guaranteed with confidence exceeding \(\frac{n+1}{2n}\).
The expected realized EF1 guarantee of every such algorithm
is also at most this threshold.

Several gaps remain. For PROP1, we leave open whether a larger asymptotic
constant is attainable.
For EF1, the optimal high-probability and
expected realized guarantees are still unknown. It would also be natural to extend
these criteria to other classical fairness notions, such as the maximin
share (MMS).

\paragraph{AI disclosure.}
GPT-6 Astra (OpenAI) was used under the authors' direction to generate
and refine proofs presented in this work. The authors independently
checked all proofs for correctness and wrote the manuscript.
The authors take full responsibility for the content and conclusions of this work.